\documentclass{article}
\usepackage[a4paper,margin=0.7in]{geometry}
\usepackage{natbib}
\usepackage{comment}
\setcitestyle{numbers}
\setcitestyle{square}
\usepackage[utf8]{inputenc} 
\usepackage[T1]{fontenc}    
\usepackage[american]{babel}
\usepackage{hyperref}       
\usepackage{url}            
\usepackage{booktabs}       
\usepackage{amsfonts}       
\usepackage{microtype}      
\usepackage{xcolor}         
\usepackage{braket}         
\usepackage{amsmath}        
\usepackage{graphicx}       
\usepackage[ruled,vlined]{algorithm2e}  

\usepackage{enumitem}       
\usepackage{amsthm}         
\usepackage{xspace}         
\usepackage{subfig}         
\usepackage{bbold}          
\usepackage{comment}

\usepackage{natbib} 
\usepackage{mathtools} 
\usepackage{booktabs} 
\usepackage{tikz} 

\newtheoremstyle{cited}%
  {3pt}
  {3pt}
  {\itshape}
  {}
  {\bfseries}
  {.}
  {.5em}
  {\thmname{#1} \thmnumber{#2} \thmnote{\normalfont#3}}

\theoremstyle{cited}
\newtheorem{citedthm}{Theorem}

\theoremstyle{plain}
\newtheorem{lemma}{Lemma}
\newtheorem{theorem}{Theorem}
\theoremstyle{definition}
\newtheorem{definition}{Definition}

\renewcommand{\P}{\mathbb{P}}

\newcommand{\R}{\mathbb{R}}

\newcommand{\E}{\mathbb{E}}

\newcommand{\1}{\mathbb{1}} 

\newcommand{\cD}{\mathcal{D}}
\newcommand{\cN}{\mathcal{N}}
\newcommand{\cU}{\mathcal{U}}
\newcommand{\cM}{\mathcal{M}}

\newcommand{\PercepOnline}{\textsc{Classical online perceptron}\xspace}

\newcommand{\QPercepOnline}{\textsc{Online quantum perceptron}\xspace}
\newcommand{\QPercepSpace}{\textsc{Version space quantum perceptron}\xspace}
\newcommand{\QPercep}{\textsc{Hybrid quantum perceptron}\xspace}

\SetKwComment{Comment}{$\vartriangleright$}{}

\SetCommentSty{mycommfont}

\SetKwInput{KwData}{Input}

\makeatletter
\newcommand\ackname{Acknowledgements}
\if@titlepage
  \newenvironment{acknowledgements}{%
      \titlepage
      \null\vfil
      \@beginparpenalty\@lowpenalty
      \begin{center}%
        \bfseries \ackname
        \@endparpenalty\@M
      \end{center}}%
     {\par\vfil\null\endtitlepage}
\else
  \newenvironment{acknowledgements}{%
      \if@twocolumn
        \section*{\abstractname}%
      \else
        \small
        \begin{center}%
          {\bfseries \ackname\vspace{-.5em}\vspace{\z@}}%
        \end{center}%
        \quotation
      \fi}
      {\if@twocolumn\else\endquotation\fi}
\fi
\makeatother

\title{Quantum Perceptron Revisited: Computational-Statistical Tradeoffs}
\author{
  Mathieu Roget, 
  Aix-Marseille Université, École Normale Superieure de Lyon, France\\
  \texttt{mathieu.roget@ens-lyon.org} \\
  Giuseppe Di Molfetta
  Aix-Marseille Université, CNRS, France\\
  \texttt{giuseppe.dimolfetta@lis-lab.fr}\\
  Hachem Kadri
  Aix-Marseille Université, CNRS, France\\
  \texttt{hachem.kadri@lis-lab.fr}
  }
 \date{}

\begin{document}

\maketitle

\begin{abstract}
Quantum machine learning algorithms could provide significant speed-ups over their classical counterparts; however, whether they could also achieve good generalization remains unclear.
Recently, two quantum perceptron models which give a quadratic improvement over the classical perceptron algorithm using Grover’s search have been proposed by Wiebe et al.~\cite{wiebe2016quantum}. While the first model reduces the complexity with respect to the size of the training set, the second one improves the bound on the number of mistakes made by the perceptron.
In this paper, we introduce a hybrid quantum-classical perceptron algorithm with lower complexity and better generalization ability than the classical perceptron. We show a quadratic improvement over the classical perceptron in both the number of samples and the margin of the data.
We derive a bound on the expected error of the hypothesis returned by our algorithm, which compares favorably to the one obtained with the classical online perceptron.
We use numerical experiments to illustrate the trade-off between computational complexity and statistical accuracy in quantum perceptron learning and discuss some of the key practical issues surrounding the implementation of quantum perceptron models into near-term quantum devices, whose practical implementation
represents a serious challenge due to inherent noise.
However, the potential benefits make correcting this worthwhile.
\end{abstract}

\section{Introduction}

\begin{table*}[t]
\centering
\caption{Summary of the computational complexities and the expected risk bounds of the classical online perceptron and the quantum perceptron models.}
\begin{tabular}{p{7cm} p{4.3cm} p{4.6cm}}
  Algorithm   &  \ Complexity & Expected risk\\
  \hline
    \PercepOnline~[e.g., \citealp{mohri2018foundations}] & \   $\displaystyle O\left(\frac{N}{\gamma^2}\right)$ & \hspace{-0cm}$\displaystyle\leq  \underset{S\sim \cD^{N+1}}{\E}\left(\frac{\min(M(S),\frac{1}{\gamma_S^2})}{N+1}\right)$\\
    
     & & \\
     
    \QPercepOnline~\citep{wiebe2016quantum} & \ $\displaystyle  O\left(\frac{\sqrt{N}}{\gamma^2}\ln\left(\frac{1}{\epsilon\gamma^2}\right)\right)$ & n/a 
    \\
    
     & & \\

    \QPercepSpace~\citep{wiebe2016quantum} & \ $\displaystyle O\left(\frac{N}{\sqrt{\gamma}}\ln^{3/2}1/\epsilon\right) $ & n/a
    \\
    
     & & \\
    
    \QPercep (this work) & \ $\displaystyle  O\!\left(\frac{\sqrt{N}}{\gamma}\ln(1/\epsilon)\ln\left(\frac{1}{\gamma\epsilon}\right)\!\right)$ & $\displaystyle\leq\underset{S\sim \cD^{N+1}}{\E}\!\left(\sqrt{\frac{\pi}{2}}\frac{\ln 1/\epsilon}{N+1}\frac{1}{\gamma_S}\right)$\\
\end{tabular}
\label{tbl:complexities}
\end{table*}

Quantum machine learning is an attractive field of research that contributes towards addressing the need for computationally efficient machine learning (ML) algorithms capable of handling huge amounts of data~\citep{wittek2014quantum,biamonte2017quantum,ciliberto2018quantum,schuld2018supervised,dunjko2020nonreview}.
Previous works in the field have mainly investigated machine learning tasks when a quantum information processing device is used, showing that a significant speed-up can be achieved compared to classical ML algorithms~\citep{rebentrost2014quantum,kerenidis2019q,kerenidis2020quantum,arunachalam2020quantum,ma2021quantum}. Quantum computing
promises the ability to solve intractable ML problems by harnessing quantum effects such as superposition and entanglement.

Quantum superposition, a fundamental concept in quantum computing, is the means by which quantum algorithms like Grover’s search can outperform classical ones. Ordinary computers operate with states built from a finite number of bits. Each bit may exist in one of the two states, 0 or 1. A quantum computer works with a finite set of objects called qubits. Each qubit has two separate states, also denoted by 0 and 1, but a qubit can also be in what is called a ``quantum superposition'' of these states, in which it is, in some sense, both 0 and 1 simultaneously~\citep{nielsen2002quantum}.
Grover’s algorithm is one of the most famous algorithm in quantum computing~\citep{grover1996fast, roget2020grover}. It solves the problem of finding one item from an unstructured database of $N$ items in time $O(\sqrt{N})$, so beating the classical $O(N)$ time requirement. Recent works have investigated the use of Grover’s search algorithm to enhance machine learning and have proved its ability of providing computational speed-up over classical ML algorithms~\citep{aimeur2013quantum, wittek2014quantum, wiebe2016quantum,li2019sublinear, casale2020quantum}.
Beyond Grover’s algorithm, quantum algorithms for linear algebra, such as quantum matrix inversion and quantum singular value decomposition, were developed and used in the context of machine learning~\citep{rebentrost2014quantum,kerenidis2017quantum}.
Among the quantum-enhanced ML algorithms that were proposed in the literature, quantum perceptron models in particular attracted our attention because it has been shown that they could enable non-trivial improvements not only in the computational complexity but also in the statistical performance of the  perceptron~\citep{wiebe2016quantum}. This may support the~(beneficial) effect of quantum computations on generalization performance.

In~\citet{wiebe2016quantum}, two quantum perceptron models based on Grover’s search algorithm were introduced. The first one (namely \QPercepOnline) is a quantum extension of the classical online perception algorithm. The complexity of the online quantum perceptron with respect to the number of examples $N$ is $O(\sqrt{N})$, which is a quadratic improvement over the classical perceptron. However, no improvement in the number of updates made by the perceptron was achieved, as its mistake bound is $O(1/\gamma^2)$, where $\gamma$ is the margin, which is the same as in the classical case.
The second quantum perceptron model (namely \QPercepSpace) is based on the notion of version space~\citep{herbrich2001bayes,mitchell1982generalization} and has a mistake bound of $O(1/\sqrt{\gamma})$, which is a substantial improvement over the classical online perceptron. Yet, as with the classical perceptron, the computational complexity of the algorithm is linear in $N$.
Hence, the question arises whether it is possible to design quantum algorithms for perceptron learning that enjoy the best features of both types of quantum perceptron models. In other words, can we develop a quantum perceptron algorithm that provides improvements in both the computational complexity and the number of mistakes the perceptron makes?

This paper provides, to the best of our knowledge, the first perceptron learning algorithm that has lower complexity and better generalization ability than the well-known classical online perceptron. Specifically, we make the following contributions: 
i) we introduce a hybrid quantum-classical perceptron algorithm (namely \QPercep) that performs a quantum search over the training set for randomly generated linear separators in order to find one that lies in the version space;
ii) we show a quadratic improvement over the classical perceptron in both the number of samples and the margin of the data; 
iii) we derive a bound on the expected error of the hypothesis returned by our algorithm that compares favorably to the one obtained with the classical online perceptron;
iv) we use numerical experiments to illustrate the trade-off between computational complexity and statistical accuracy in quantum perceptron learning and discuss some of the key practical issues surrounding the implementation of quantum perceptron models into near-term quantum devices, whose practical implementation represents a serious challenge due to inherent noise.
Our theoretical results for Quantum Perceptron  and other related works are summarized in Table~\ref{tbl:complexities}.

\section{Preliminaries}
\label{sec:preliminaries}

We begin with reviewing the classical perceptron algorithm and then give some background on quantum computing and Grover’s search algorithm.

\subsection{Classical perceptron algorithm}
\paragraph*{Algorithm and complexity}

The perceptron is an online algorithm designed to solve binary classification problems~\citep{rosenblatt1958perceptron}. It has received a lot of attention due to its simplicity and versatility~\citep{cesa2005second,freund1999large,shalev2005new}. 
Consider a training set $\{(x_1,y_1), . . . ,(x_N,y_N)\}$ with data vectors $x_i \in \R^D$ and class labels $y_i \in \{-1,1\}$, $i=1,\ldots, N$.
We assume that the data are linearly separable, i.e., there exists a hyperplane that separates the data points of the class $1$ from those of the class $-1$. The \PercepOnline will find a separator $w\in \R^D$ such that $\forall i, \; y_i w^T x_i\geq 0$. The algorithm simply updates the vector $w$ each times it misclassifies a point. The \PercepOnline is depicted in Algorithm \ref{classical_perceptron}. 

The margin $\gamma$ between the two classes is defined by:
\begin{equation*}
    \gamma = \max_{v\in\R^D}\min_{1\leq i\leq N}\frac{y_i \langle v, x_i\rangle }{\Vert v\Vert}.
\end{equation*}
Usually, the margin is small (close to zero) which means that the classes are close and separating them is hard. In the following, we always assume that the margin is smaller than one (which can be achieved by normalizing the training set) and the asymptotic complexities are studied when $N$ and $\frac{1}{\gamma}$ are large. 
When the norm of the $x_i$'s is at most $1$, it holds that the number of updates made by the perceptron during the learning phase is at most $O(\frac{1}{\gamma^2})$. This result is known as the bound of Novikoff \citep{novikoff1962convergence, mohri2013perceptron}. If we want to correctly classify all the $N$ samples, the final complexity of the perceptron is then $O(\frac{N}{\gamma^2})$.

\paragraph*{Generalization}

One of the most fundamental questions in Machine Learning is what are the generalization guarantees of a learning algorithm.
The perceptron algorithm learns a mapping
between input data and target labels using a finite sample of labeled examples, and then uses a hyperplane to separate the data and predict the class of unseen examples.
It is therefore important to assess the ability of the perceptron to generalize to unseen data.
In a statistical learning framework, such assessment is often performed by bounding the risk (or generalization error).
Let us denote by $\cD$ the distribution that 
generates the data. The training sample $S$ of $N$ data points $ \{(x_i,y_i)_{i=1}^N\}$ is assumed to be drawn randomly from the (unknown) distribution $\cD$ and we write $S \sim \cD^N$. The binary classification risk is defined by
%
%
%
%
 \begin{equation*}
    \label{eq:risk}
    R(h_S) = \underset{(x,y)\sim \cD}{\E}\left(\1\{h_{S}(x)\neq y\}\right)\; ,
\end{equation*}
where $h_S$ is the hypothesis returned by the algorithm on the sample $S$. 

\begin{citedthm}
\label{th:generalization_perceptron}
Assume that the data are linearly separable. Let $h_S$ be the hypothesis returned by the
\PercepOnline algorithm after training over a sample $S$ of size $N$ drawn according to some distribution $\cD$. We note $\gamma_S$ the margin of sample $S$. Then, the expected risk of $h_S$ is bounded as follows:
\begin{equation*}
    \E_{S\sim \cD^N}\left(R(h_S)\right) \leq \frac{1}{N+1}\E_{S\sim \cD^{N+1}}\left(\min(M(S),{1}/{\gamma_S^2})\right),
\end{equation*}
where $M(S)$ is the number of updates made by the algorithm after training over $S$.
\end{citedthm}
\begin{proof}
See \citep[Theorem 8.9]{mohri2018foundations}.
\end{proof}


\begin{algorithm}[ht]
\RestyleAlgo{ruled}
\caption{\PercepOnline}
\label{classical_perceptron}

\KwData{data $(x_i,y_i)_{1\leq i\leq N}$ \tcp*[r]{training set}}

    $w \gets 0$ \tcp*[r]{separator in $\R^D$}
    \While{$(x_t,y_t) \gets \textsc{Receive}()$ \tcp*[r]{data we receive}}{ 
            \If{$y_t w^T x_t\leq 0$\tcp*[r]{data wrongly classified}} {
                $w \gets w + y_t x_t$    \tcp*[r]{we update the separator}
            }
    }
    \Return{$w$} 

\end{algorithm}


\subsection{Quantum computation and Grover's search algorithm}

Before we introduce the quantum perceptron algorithm, we believe it is opportune to briefly present the principles of quantum mechanics, i.e. the underlying mathematical structure of all quantum physical systems. It is not possible to provide here a complete and exhaustive presentation, so we will limit ourselves to introduce only those “game rules” useful to understand the content of the following algorihtms, leaving it to the reader’s curiosity to more complete reviews, such as \citet{nielsen2002quantum}. At best, we first introduce the arena where the game goes on, then we define the dynamics of the quantum system and finally we shortly introduce the measurement operation. 

While classically a computational state takes value in $\{0,1\}$, a quantum state is represented by a unit complex vector $\ket{\psi}$ in the Hilbert space $\mathbb{C}^2$. Such state space is equipped by an orthonormal basis $\{\ket{0},\ket{1}\}$, such that any vector is generally described by a convex linear combination 
\begin{equation*}
\ket{\psi} = \alpha_0 \ket{0} + \alpha_1 \ket{1},
\end{equation*}
where $(\alpha_i)_{i=0,1}$ are complex numbers. More in general qubit basis states can also be combined to form product basis states to describe multi-qubits systems. If $\ket{\psi_1},\ket{\psi_2},...,\ket{\psi_n}$ represent the states of $n$ isolated quantum systems, the state of the composite system is given by the tensor product of the state space of the components : $\ket{\psi_1}\otimes \ket{\psi_2}\otimes...\otimes\ket{\psi_n}$. A concrete example of composite system is the memory of a $n-$qubit quantum computer, where each qubit is called register. In that case, \begin{equation*}
    \ket{\psi} = \sum_{i=0}^{2^n-1}{\alpha_i\ket{i}}.
\end{equation*} 
Similarly to classical computing, we can act by means of logical gates onto such quantum register to perform computation. Quantum circuits are nothing but reversible logical circuits onto complex-valued state space. Each quantum gate requires a special kind of reversible function, namely a unitary mapping, that is, a linear transformation of a complex inner product space that preserves the Hermitian inner product. When such systems are kept isolated, the computation is kept reversible. However we need to obtain classical information about the outcome of a quantum computation task. In practice a quantum state has to be measured which formally coincides with an orthogonal projector onto one of the computational basis state $\ket{v}\in \mathbb{C}^n$. During such measurement operation, the quantum state is randomly collapsed into a classical state, with probability $\displaystyle \P\left(v=i \mid v \overset{\text{Meas}}{\longleftarrow} \psi\right) = |\alpha_i|^2, \; \forall \ 0\leq i < 2^n$, where $v$ has been expressed in a decimal system.


\begin{algorithm}[t]
\caption{QSearch}
\label{qsearch}
\RestyleAlgo{ruled}
\KwData{data $\{x_i\}_{1\leq i\leq N}$ \tcp*[r]{data we want to search in}}
\KwData{oracle $f$ \tcp*[r]{oracle such that $f(x_i) = \1\{i\in \cM\}$}}

    $\displaystyle \psi_0 \gets \frac{1}{\sqrt{N}}\sum_{i=1}^N\ket{i}$\\
    $R \gets \textsc{Quantify}(f)$ \tcp*[f]{quantum version of the oracle}\\
    $U_g \gets GR$\\
    $\displaystyle m\gets \cU\left(\left\{0,\ldots,\left\lceil \frac{1}{\sin(2\sin^{-1}\left(\sqrt{\frac{1}{N}}\right))} \right\rceil -1\right\}\right)$\\
    $v \overset{\text{Meas}}{\longleftarrow} U_{g}^m\psi_0$\\
    \Return{$v$} 

\end{algorithm}


At the heart of the quantum perceptron algorithm lies the quantum search algorithm, which is widely used as main routine in many algorithms, generally guaranteeing to speed up any brute force $O(N)$ problem into a $O(\sqrt{N})$ problem. It has been introduced by~\citet{grover1996fast} as a fast quantum mechanical algorithm for database search algorithm and it represents one of the most important and studied algorithm in quantum computing. In the following, we shortly present the Grover algorithm.
Let us consider $N = 2^n$ elements and $\cM \subseteq \{1,\ldots,N\}$ the searched elements. We start with the diagonal quantum state 
\begin{equation*}
    \ket{\psi_0} = \frac{1}{\sqrt{N}}\sum_{i=1}^N\ket{i}\; .
\end{equation*}
We then apply two operators: an oracle and a reflection. The oracle $R$ is defined by
\begin{equation*}
    R\ket{x} = \left\{\begin{matrix} -\ket{x} & \text{if } x\in \cM\\ \ket{x} & \text{otherwise,}\\\end{matrix}\right.
\end{equation*}
while the reflection $G$ is given by
\begin{equation*}
    G = 2\psi \psi^\dagger - \1 \; .
\end{equation*}
This two operators can in fact be view in a geometric way. We note $\#\cM$ the cardinal of $\cM$. Let's denote $a = \frac{\#\cM}{N}$ the probability to find a searched element before running the algorithm (when the state is diagonal) and $\theta_a = \sin^{-1}\left(\sqrt{a}\right)$, the angle between the subspace composed by the searched elements and the complementary subspace. Then one can show that $U_{g} := GR$ is a rotation of an angle $2\theta_a$, meaning that after $j$ steps the probability to measure a searched element is
\begin{equation*}
    \P\left(v\in \cM \mid v \overset{\text{Meas}}{\longleftarrow} U_{g}^j\psi_0\right) = \sin^2\big((2j+1)\theta_a \big)\; .
\end{equation*}
It is then easy to find the number of steps that gives the optimal probability of finding a searched element. But to find this optimal number of steps, one needs to know $\theta_a$ which is directly related to the number of searched elements. We want here to adapt this algorithm in order to make it for an unknown number of searched elements. 

The idea here that comes from~\citet{boyer1998tight} is simply to draw the number of steps randomly uniformly between $0$ and $M-1$. The resulting probability is 
\begin{align*}
    &\P\left(v\in \cM \mid v \overset{\text{Meas}}{\longleftarrow} U_{g}^m\psi_0, m\gets \cU_{\{0,\ldots,M-1\}}\right) \\  &=\frac{1}{M}\sum_{j=0}^{M-1}{\sin^2\big((2j+1)\theta_a \big)}
     = \frac{1}{2}\left(1-\frac{\sin(4M\theta_a)}{2M\sin(2\theta_a)}\right)\; .
\end{align*}
If $M\geq \frac{1}{\sin(2\theta_a)}$, then it holds that this probability is at least $\frac{1}{4}$. The last thing we need is to express a bound for $M$ that doesn't depend on $\theta_a$:
\begin{align*}
    M \geq \frac{1}{\sin(2\theta_a)} &=\frac{1}{\sin(2\sin^{-1}\left(\sqrt{\frac{\#\cM}{N}}\right))}\\
    &\leq \frac{1}{\sin(2\sin^{-1}\left(\sqrt{\frac{1}{N}}\right))}  = O(\sqrt{N})\; .
\end{align*}
In other words, we bound $M$ by its maximum value which occurs when $\#\cM=1$ (i.e. one marked element).
The detailed quantum search over an unknown number of searched elements is given in Algorithm~\ref{qsearch}. This algorithm find a searched element with probability at least $\frac{1}{4}$ and has a complexity $O(\sqrt{N})$.
By repeating the algorithm a logarithmic number of times, we can increase the probability of success to $1-\epsilon$ for any $\epsilon > 0$~\citep{wiebe2016quantum}.

\section{Existing quantum perceptron algorithms}

In this section, we discuss two existing quantum perceptron algorithms proposed in~\citet{wiebe2016quantum} that are closely related to our work.
Note that other quantum perceptron models can  be  found  in  the  literature  of  quantum  neural  networks~\citep{behrman2000simulations,ricks2003training,schuld2015simulating}.

\subsection{Online quantum perceptron}

\begin{algorithm}[t]
\caption{\QPercepOnline~\citep{wiebe2016quantum}}
\label{alg:online_quantum_perceptron}
\RestyleAlgo{ruled}
\KwData{data $(x_i,y_i)_{1\leq i\leq N}$ \tcp*[r]{training set}}

    $w \gets 0$ \tcp*[r]{separator in $\R^D$}
    \For{$i\in \{1,\ldots,1/\gamma^2\}$ \tcp*[r]{we perform enough updates}} {
        \For{$j\in \{1,\ldots,\left\lceil\log_{3/4}(\gamma^2\epsilon)\right\rceil\}$ \tcp*[r]{we increase the probability of \textsc{QSearch}}} {
            $m \gets \textsc{QSearch}(\{(x_k,y_k)\}_k)$ \tcp*[r]{searching for a point $x_m$ misclassified by $w$}
            \If{$y_m w_i^T x_m \leq 0$ \tcp*[r]{If actually misclassified\ldots}} {
                $w \gets w + y_m x_m$ \tcp*[r]{... then update}
            }
        }
    }
    \Return{$w$} 

\end{algorithm}

The classical online Perceptron updates the hyperplane when an example is misclassified  and stops when all training data are correctly classified.
The online quantum perceptron works similarly to the classical one. The main difference is the means by which misclassified points are detected.
Instead of testing each point one by one, a Grover search is
performed to find a wrongly classified example. Once this is done, the hyperplane is updated and the process is repeated until convergence.
the \QPercepOnline is outlined in Algorithm~\ref{alg:online_quantum_perceptron}.
Note that this algorithm is not really an online algorithm since it considers a quantum superposition of states representing the training data samples.
The naming ‘online' quantum perceptron is used because this algorithm has the same update rule than the classical online perceptron. 
In this quantum version of the perceptron, the computational complexity is improved from $O(N)$ to $O(\sqrt{N})$ due to the Grover search. However, an additional $\log \left(1/(\epsilon\gamma^2)\right)$ will appear to deal with the probability of failure of the quantum search. This is summarized in the theorem below.

\begin{citedthm}[\citep{wiebe2016quantum}]
Let $S$ be a linearly separable sample of $N$ points of margin $\gamma$. Algorithm \QPercepOnline finds a perfect separator with probability at least $1-\epsilon$ and has a complexity of
\begin{equation*}
    O\left(\frac{\sqrt{N}}{\gamma^2}\log\left(\frac{1}{\epsilon\gamma^2}\right)\right)\; .
\end{equation*}
\end{citedthm}

\subsection{Version space quantum perceptron}

The idea of the second quantum perceptron model is based on the notion of version space, which is the set of hypotheses that are consistent with the training data~\citep{herbrich2001bayes}.
Here, $K$ linear separators are randomly drawn from the normal distribution $\cN(0,\1)$, so the problem becomes how to find one of these separators that is in the version space, i.e., correctly separates the data.
Using a version space point of view, the perceptron learning problem is transformed into a search problem and then quantum search algorithms can be used to solve it efficiently.
 The Grover search is now applied over the generated hyperplanes and not the training set as in the previous algorithm (see Algorithm~\ref{alg:version_space_quantum_perceptron}).
 A significant improvement on the number of hyperplanes $K$ is achieved; however, a full pass over the training examples is needed to find the hyperplane that belongs to the version space.
 The computational complexity of the algorithm is $O(N\sqrt{K})$ with an additional $\log{1/\epsilon}$ because of the probability of failure, as summarized in the theorem below.
 Note that~\citet{wiebe2016quantum} provided a result about the number of hyperplanes that must be generated to guarantee that at least one of them is in the version space.  Interestingly, this number depends on the margin of the data.
Indeed, it was shown that the number of hyperplanes to be sampled is $K = O\left(\frac{\ln(1/\epsilon)}{\gamma}\right)$.

\begin{algorithm}[t]
\caption{\QPercepSpace~\citep{wiebe2016quantum}}
\label{alg:version_space_quantum_perceptron}
\RestyleAlgo{ruled}
\KwData{data $(x_i,y_i)_{1\leq i\leq N}$ \tcp*[r]{training set}}

    Draw $\{w_1,\ldots,w_K\} \gets \cN(0,\1)$ \tcp*[r]{We assume this is done efficiently}
    \For{$i\in \{1,\ldots,\left\lceil\log_{3/4}(\epsilon)\right\rceil\}$ \tcp*[r]{we increase the probability of \textsc{QSearch}}} {
        $m \gets \textsc{QSearch}(\{w_k\}_{k})$ \tcp*[r]{searching $\!$for$\!$ a$\!$ separator$\!$ $w_m$ that$\!$ correctly$\!$ classifies$\!$ the$\!$ data}
        \If{$y_j w_m^T x_j > 0, \; \forall j$ \tcp*[r]{check if the obtained hyperplane is a good one}} {
            \Return{$w_m$}
        }
    }
    \Return{$w_{1}$} 

\end{algorithm}

\begin{citedthm}[\citep{wiebe2016quantum}]
Let $S$ be a linearly separable sample of $N$ points of margin $\gamma$. Algorithm \QPercepSpace finds a perfect separator with probability at least $1-\epsilon$ and has a complexity of
\begin{equation*}
    O\left(\frac{N}{\sqrt{\gamma}}\log^{3/2}\left(\frac{1}{\epsilon}\right)\right)\; .
\end{equation*}
\end{citedthm}
As we can see, this algorithm does not improve the complexity with respect to the number of the training data $N$; but it has a better statistical guarantee than the classical perceptron, since the classical mistake bound of $O(1/\gamma^2)$ can be improved to $O(1/\sqrt{\gamma})$.
In the next section we propose a quantum perceptron algorithm that has the two advantages of the online and the version space quantum perceptron: it provides improvements in both the computational complexity and the number of mistakes.

\section{Hybrid quantum perceptron: an improved perceptron learning}

This section presents our main results.
We introduce a hybrid quantum perceptron algorithm to take advantage of
the two quantum perceptron models described above.
We show a quadratic improvement over the classical perceptron in both the number of samples and the margin of the data. 
Then, we derive a bound on the expected error of the hypothesis returned by our algorithm.

\subsection{Algorithm}

 The idea is also to draw randomly several linear separators following the normal distribution $\cN(0,\1)$ and then search for one in the version space, so it correctly separates the data. However, in contrast to the \QPercepSpace, our algorithm will perform a quantum search over the training set for each separator to find a solution, and not a quantum search over the separators.
 By doing this, we can improve the complexity with respect to the number of samples $N$, as for the \QPercepOnline, while still enjoying the benefits
of the version space approach. 
 Our hybrid quantum perceptron algorithm is described in Algorithm~\ref{alg:hybrid_quantum_perceptron}.

\begin{citedthm}
\label{th:hybrid_quantum_perceptron}
Let $S$ be a linearly separable sample of $N$ points of margin $\gamma$. Algorithm \QPercep finds a perfect separator with probability at least $1-\epsilon$ and has a complexity of
\begin{equation*}
     O\left(\frac{\sqrt{N}}{\gamma}\ln(1/\epsilon)\ln\left(\frac{1}{\gamma\epsilon}\right)\right)\; .
\end{equation*}
\end{citedthm}
\begin{proof}
See supplementary materials.
\end{proof}
This is a quadratic improvement in the computational and statistical complexity of the classical online perceptron. The improvement of the statistical complexity is quadratic only if we assume that the data supplied to the classical perceptron are provided the same way that the quantum one.
Indeed, the complexity of the classical perceptron in this case is $O(\frac{N}{\gamma^2}\log(\frac{1}{\epsilon\gamma^2}))$ (see~\citet[Th.~1]{wiebe2016quantum}). 
If the classical perceptron is online instead, then the statistical complexity improve from $O((1/\gamma)^2)$ to $O(1/\gamma \ln(1/\gamma))$ which is slightly less than quadratic.
The computational improvement is due to the quantum search while the statistical improvement is provided by our choice of using a version space based strategy, leading to the name ‘hybrid QP'.
Theorem~\ref{th:hybrid_quantum_perceptron} shows that our algorithm is particularly well-suited for large-scale data sets and small margins.

\begin{algorithm}[t]
\caption{\QPercep}
\label{alg:hybrid_quantum_perceptron}
\RestyleAlgo{ruled}
\KwData{data $(x_i,y_i)_{1\leq i\leq N}$ \tcp*[r]{training set}}
\KwData{$\{w_1,\ldots,w_K\} \sim \cN(0,\1)$ \tcp*[r]{hyperplanes}} 

    \For{$i\in \{1,\ldots,K\}$ \tcp*[r]{for all hyperplanes\ldots}} {
        $b \gets 1$ 
        \For{$j\in \{1,\ldots,\left\lceil\log_{3/4}\left(1-\left(1-\frac{\epsilon}{2}\right)^{\frac{1}{K-1}}\right)\right\rceil\}$ \tcp*[r]{increase \textsc{QSearch} success probability}} {
            $m \gets \textsc{QSearch}(\{(x_k,y_k)\}_k)$  \tcp*[r]{searching for a point $x_m$ misclassified by $w_i$}
            \If{$y_m w_i^T x_m \leq 0$ \tcp*[r]{if one is found\ldots}} {
                $b \gets 0$ \tcp*[r]{\ldots then the current hyperplane isn't a good one}
            }
        }
        \If{$b = 1$ \tcp*[r]{if no miclassified point found\ldots }} 
        {\Return{$w_i$} \tcp*[r]{\ldots then return the current hyperplane}
        }
    }
    \Return{$w_{1}$} 

\end{algorithm}

\subsection{Generalization}

In the classical setting, mistake bounds for the Perceptron algorithm  can be used to derive generalization bounds~\citep{cesa2004generalization,mohri2013perceptron}.
This question was not addressed in~\citet{wiebe2016quantum}.
As we have seen above, the \QPercep provides an improvement on the statistical efficiency of the perceptron~($O(1/\gamma)$ instead of  $O(1/\gamma^2)$). We show here that this may yield better generalization guarantees.

 We have a training set $S=\{z_1,\ldots,z_n\}$ with $z_i=(x_i,y_i)$. We assume that $z_i$ are independently sampled from an unknown distribution $\cD$.
We recall that the risk is defined by 
\begin{equation*}
    R(h) = \E_{z\sim \cD} \left(\1\{h(x)\neq y\}\right)\; ,
\end{equation*}
where $h$ is a hypothesis in a hypothesis set $\mathcal{H}$.

\begin{citedthm}
\label{th:generalization_quantum_perceptron}
Assume that the data is linearly separable. Let $h_S$ be the hypothesis returned by the
\QPercep algorithm after training over a sample $S$ of size $N$ drawn according to some distribution $\cD$. Then, the expected error of $h_S$ is bounded as follows:
\begin{equation*}
    \E_{S\sim \cD^N}\left(R(h_S)\right) \leq \sqrt{\frac{\pi}{2}}\frac{\log 1/\epsilon}{N+1}\E_{S\sim \cD^{N+1}}\left(\frac{1}{\gamma_S}\right)\; .
\end{equation*}
\end{citedthm}
\begin{proof}
See supplementary materials.
\end{proof}

\begin{figure*}[t]
    \centering
    
    \subfloat[]{%
    \includegraphics[width=0.35\textwidth]{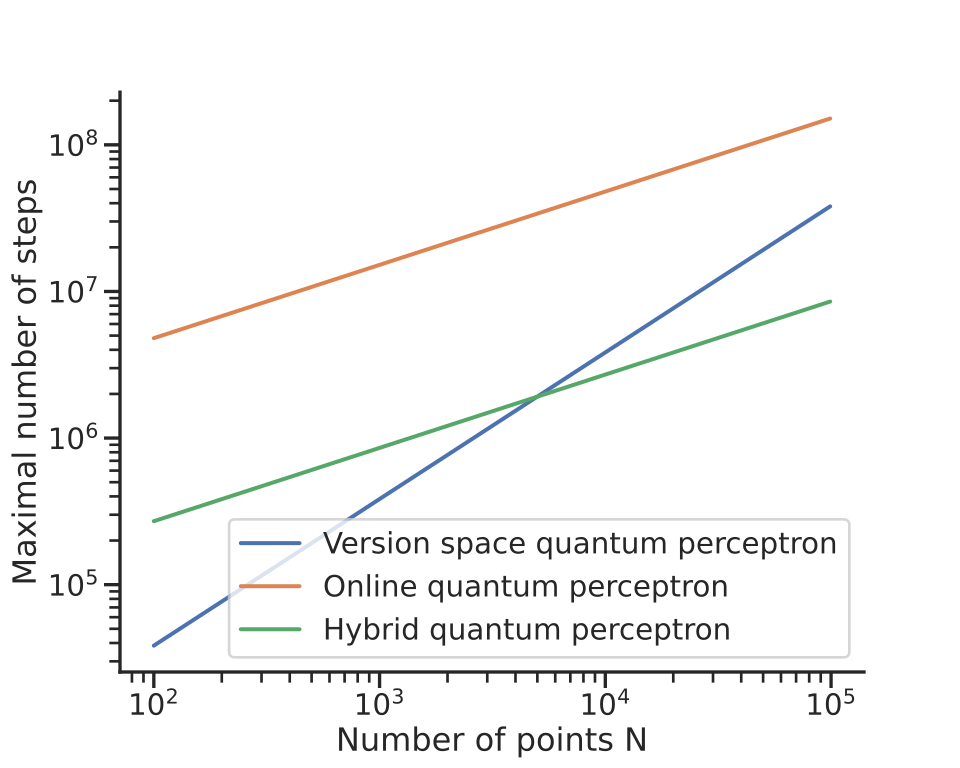}%
    \label{complexities:1}%
    }%
    \qquad\qquad\qquad
    \subfloat[]{%
    \includegraphics[width=0.35\textwidth]{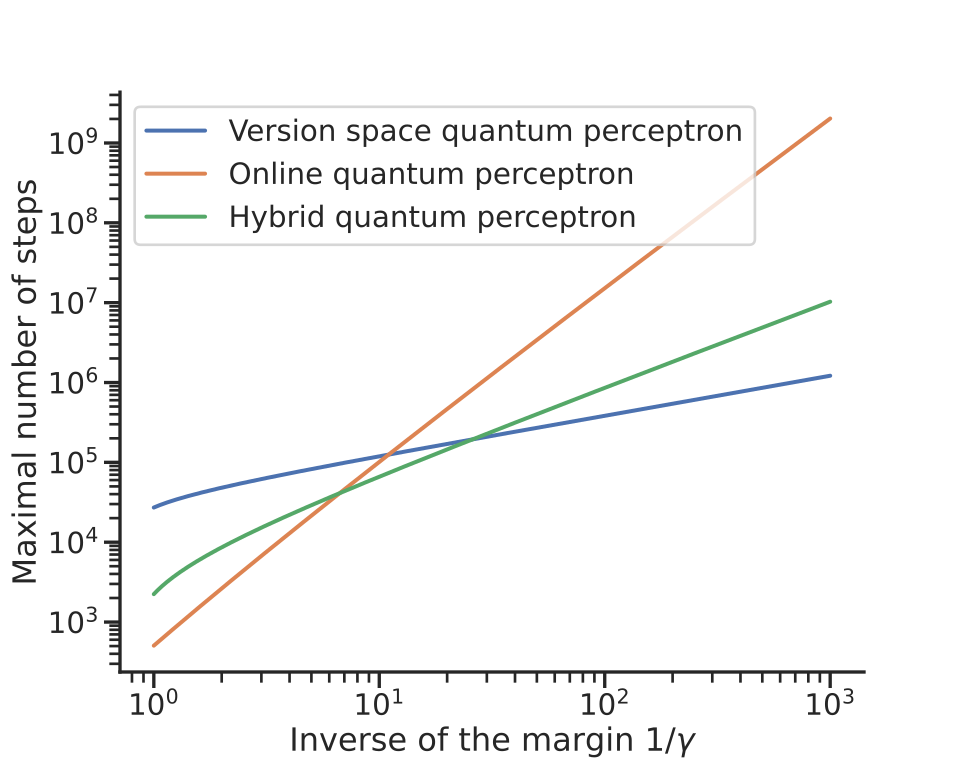}%
    \label{complexities:2}%
    }%
    \caption{Complexity bounds over the number of operations for quantum perceptrons. The curves asymptotically follow the complexities summarized in Table~\ref{tbl:complexities}. Subplot (a) 
    shows the bounds in function of the number of points $N$ with a margin $\gamma = 0.01$. Subplot (b) 
    shows the bounds in function of the inverse of the margin $\gamma$ with a number of points $N=1000$.}
    \label{complexities}
\end{figure*}

The bound obtained in the classical online setting is equal to $\frac{1}{N+1}\E_{S\sim \cD^{N+1}}\left(\min(M(S),{1}/{\gamma_S^2})\right)$, where $M(S)$ is the number of updates made by the algorithm after training over $S$~\citep[Theorem 8.9]{mohri2018foundations}.
Theorem~\ref{th:generalization_quantum_perceptron} shows that \QPercep can give considerable improvement in generalization over the classical online perceptron algorithm. 
However, the guarantee given is not  a high probability bound, since it holds only for the expected error of the hypothesis returned by the algorithm.

\section{Numerical experiments}

In this section, we illustrate empirically the theoretical performance guarantees introduced in the previous section. Then , we discuss the effect of quantum noise  which is one of the major issue of near-term quantum algorithms. The simulations presented here come from a classical computer simulating a quantum algorithm.\footnote{The code to reproduce our experiments is available in a GitHub repository: \url{https://github.com/mroget/Quantum-perceptron-models}.}

\subsection{Computational-statistical trade-off}\label{sec:Computational-statistical trade-off}

We run experiments with the three quantum perceptron models studied in this paper and compare the number of steps required for these algorithms when varying the number of data samples $N$ and the margin $\gamma$.
Figure \ref{complexities} shows the maximal number of steps; namely the complexities taking into account the constants.
 The slope of the curves gives an indication of the complexity in terms of $N$ or $1/\gamma$, while the intercept provides a good indication of the impact of the constant factors on it.
 The slope of the curve of \QPercep is lower than the one of the \QPercepSpace when $\gamma$ is fixed and $N$ varies and also lower than the slope of \QPercepOnline when $N$ is fixed and $\gamma$ varies. This confirms that our algorithm has a lower computational complexity and also a better statistical efficiency.

It is also interesting to compare the behavior of these quantum perceptron algorithms with respect to the number of operations made by the classical online perceptron.
We apply the three quantum perceptron algorithms on the Iris dataset and on a simulated dataset (called Hard). Iris is a simple dataset for which the classical perceptron will converge very quickly. The Hard dataset, however,  is specifically build to force the classical perceptron algorithm to perform a large number of updates. 
\begin{definition}[Hard dataset]

The Hard dataset inspired from \citet[Exercice 8.1]{mohri2018foundations} is composed of a sample $S_H(N) = \{(x_1,y_1),\ldots,(x_N,y_N) \in (\R^N\times \{0,1\})^N$ of size $N$ and dimension $N$ such that
\begin{equation*}
    \forall i,j \in [N]^2,  \; (x_i)_j = (-1)^{i+1}\1\{j=i\} \text{ and } y_i = (-1)^i\; .
\end{equation*}
\end{definition}
Figure~\ref{perf} shows the ratio between the number of operations of each quantum perceptron algorithm and the number of steps of the classical perceptron during the learning phase. 
On the Iris dataset, the three quantum perceptrons behave similarly and are about four times slower than the classical perceptron. This is expected since the problem is easy to solve. For the Hard dataset, however, all the quantum perceptron algorithms shows an improvement over the classical one. Interestingly, \QPercep is the one that performs the best, since it achieves
a good trade-off between computational and statistical complexities.

\begin{figure}[t]
    \centering
    \includegraphics[width=0.53\textwidth]{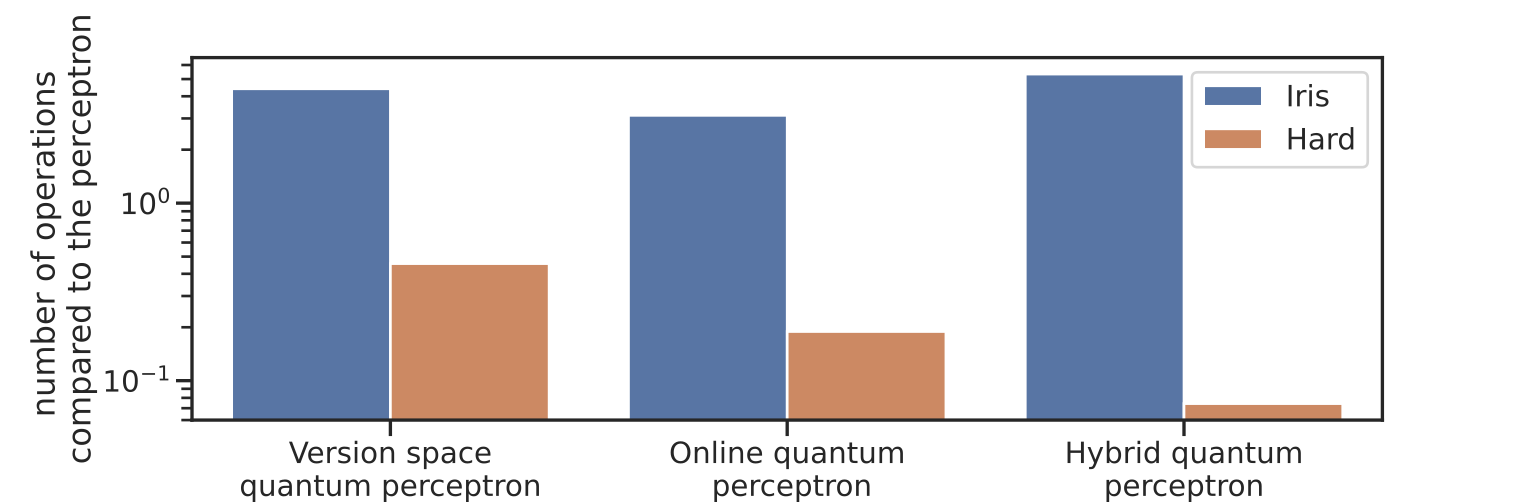}
    \caption{Ratio between the number of operations of quantum perceptron and  classical perceptron.}
    \label{perf}
\end{figure}

\subsection{Quantum noise}\label{sec:Quantum noise}

Most of the existing quantum devices are subject to quantum noise. 
Dealing with noise in quantum computation is nowadays an important and challenging problem.
Although a rigorous analysis goes beyond the scope of this work, here we shortly illustrate how noise may affect the quantum perceptron computational task. 
All quantum algorithms presented in this paper are based on the assumption that the quantum search finds a searched element with probability at least $1/4$. As a reminder (see Section~\ref{sec:preliminaries}), the quantum search is designed by performing $m$ steps of the Grover's algorithm while $m$ is drawn uniformly between $0$ and $M$. We can compute the probability of success of the quantum search with respect to $M$. Let us call this probability $P(M)$. As seen in Section~\ref{sec:preliminaries}, it holds that
\begin{equation*}
    P(M) = \frac{1}{2}\left(\frac{\sin(4M\theta_a)}{2M\sin(2\theta_a)}\right)\;.
\end{equation*}
Here, $\theta_a$ depends only on the proportion of searched elements. Figure \ref{p(m)} shows three curves. Each one is the plot of $P(M)$ for one searched element with a specific quantum noise model. The blue one does not account noise while the other two curves have, respectively, bit-flip, and depolarization noise~\citep{wang2020prospect}. The first class of noise coincides with a unitary random flip, meaning that the computational state flips from $\ket{1}$ to $\ket{0}$ or vice versa. The second kind of error can be seen as a completely positive trace-preserving map from the quantum state onto a linear combination of itself and a general maximally mixed state. 
As we can see, the success probability in a fault-free environment converges towards $1/2$, thus for large enough $M$, $P(M)$ can be always greater than $1/4$, as explained in the section Preliminaries.
However in a faulty-environment, the success probability decreases rapidly and do not tend to a non vanishing constant, making harder to recover a $P(M)$ greater than $1/4$. Moreover the quantum noise strictly depends on the quantum circuit design (in concrete how errors may propagate), making this choice crucial to build a fault-tolerant quantum perceptron.
This decreasing is the result of making too many iterations, thus accumulating noise. On the other hand, the probability starts by increasing because the quantum search is working. The peak of probability represents the best trade-off between the increase of the probability of success and 
the increase of the quantum~noise.

\begin{figure}[t]
    \centering
    \includegraphics[width=0.43\textwidth]{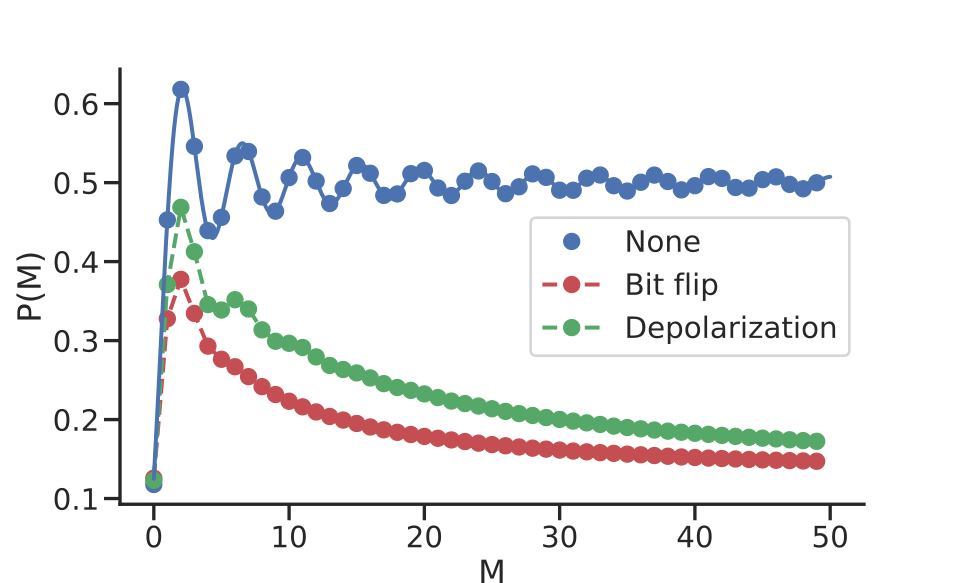}
    \caption{$P(M)$ for different noise models.} 
    \label{p(m)}
\end{figure}

\section{Discussion}

In this work, by classical perceptron we mean the standard online perceptron. There is, to our knowledge, no mention in the classical ML literature to classical version space perceptron. The Quantum Perceptron algorithm we propose has a quadratic improvement in $N$ and $\gamma$ over the well-known classical online perceptron. Similarly, a quartic speed-up is used in~\citet{wiebe2016quantum} to describe the improvement over $\gamma$ they obtained with their quantum version space perceptron.
It is worth noting that, although it is not known in the literature, a classical version space perceptron should have a complexity bound inversely proportional to the margin $\gamma$. The quadratic improvement over the margin is not provided by the Grover's search algorithm but by the version space approach. 
Usually the version space approach scales linearly with the number of examples $N$. The quadratic improvement in $N$ is, however, obtained by our quantum perceptron using a quantum search over the training set.
When adopting a version space approach, the perceptron problem is transformed into a search problem over the generated hyperplanes. Our results show that applying a quantum search over the training set and not over the hyperplanes in this situation provides new insights for the design of computationally and statistically efficient perceptron models.

To our knowledge, our Theorem~\ref{th:generalization_quantum_perceptron} is the first result showing that the version space perceptron (classic or quantum) can have a better generalization than the online perceptron algorithm.
There are no results concerning the expected risk of previous quantum perceptron algorithms.
%
We expect that the expected risk bound of \QPercepOnline is of the same order than the classical online perceptron, since this algorithm does not improve the mistake bound.
For \QPercepSpace, it is not clear whether the improvement on the scaling of the  algorithm with respect to the margin could yield even better generalization guarantees. The factor $1/\gamma$ in the expected risk bound of \QPercep is related to the number of  the randomly generated linear separators~(see the proof of Theorem~\ref{th:generalization_quantum_perceptron}). The version space quantum perceptron has the same number of separators than our algorithm. So, using the same line of proof as for Theorem~\ref{th:generalization_quantum_perceptron} will not necessarily result in an improved bound.

In this paper we only considered linear classification. In the classical case, kernel methods provide a powerful tool for generalizing linear classifiers to nonlinear settings~\citep{scholkopf2002learning}.
With appropriate nonlinear features, linear models can be used to approximate a nonlinear function. Kernel methods allow the construction of these nonlinear features.
There are interesting links between kernel methods and quantum computing~\citep{havlivcek2019supervised,schuld2019quantum}. Indeed, the process of encoding inputs in a quantum state can be interpreted as a nonlinear feature map that maps data to a quantum Hilbert space. So, the quantum encoding of classical data can be seen as a way to construct nonlinear quantum features. Different quantum encodings were proposed and the corresponding kernels were given.
Nonlinear extensions of our work can be carried out by the classical-to-quantum data encoding scheme.

\section{Conclusion}
In this paper, we proposed a hybrid quantum perceptron algorithm that goes beyond  the ideas of existing quantum perceptron algorithms. This allowed us to obtain a quadratic improvement over the computational complexity and the statistical efficiency compared to the classical online perceptron. 
We performed numerical experiments to  support our theoretical findings.
In the future, it would be valuable to study noise-robust models for quantum perceptron.

\begin{acknowledgements}
We thank L. Ralaivola for useful discussions. This work has been funded by the French National Research Agency~(ANR) project QuantML~(grant number ANR-19-CE23-0011) and the INS2I CNRS project QuAlgo. 
\end{acknowledgements}



\appendix
\setcounter{theorem}{3}

\section{Proofs}\label{sec:proofs}
In this appendix, we present the proofs of Theorems \ref{proof:hybrid_quantum_perceptron} and \ref{proof:generalization_quantum_perceptron}.

\subsection{Proof of Theorem \ref{proof:hybrid_quantum_perceptron}}
\label{sec:proof1}

After proving a few useful lemma, we provide here the proof of the complexity of our \QPercep.

\begin{lemma}\label{lemma:K}
Let's define 
$K = \left\lceil\frac{\ln(\epsilon/2)}{\ln(1-\sqrt{2}\gamma/\sqrt{\pi})}\right\rceil$, then it holds that 
$$ K \sim \sqrt{\frac{\pi}{2}} \frac{\ln(1/\epsilon)}{\gamma}. $$
\end{lemma}
\begin{proof}
Using a Taylor expansion for $\ln(1-x)$ in 0 we get

\begin{align*}
    \sqrt{\pi/2}\frac{\ln(1/\epsilon)}{K\gamma} &= \sqrt{\pi/2}\frac{\ln(1/\epsilon)\ln(1-\sqrt{2}\gamma/\sqrt{\pi})}{\gamma\ln(\epsilon/2)}\\
    &=  \sqrt{\pi/2}\frac{\ln(1/\epsilon)\left[-\sqrt{2}\gamma/\sqrt{\pi} + \underset{\gamma \to 0}{o}(\gamma)\right]}{\gamma\ln(\epsilon/2)}\\
    &\underset{\gamma \to 0}{\to} \frac{\ln(1/\epsilon)}{\ln(1/\epsilon) + \ln(2)}\\ &\underset{\epsilon \to 0}{\to} 1.
\end{align*}

Thus $K \sim \sqrt{\pi/2}\frac{\ln(1/\epsilon)}{\gamma}$.
\end{proof}

\begin{lemma}\label{lemma:K2}
Let's define 
$K2 = \left\lceil\log_{3/4}\left(1-\left(1-\frac{\epsilon}{2}\right)^{\frac{1}{K-1}}\right)\right\rceil$, then it holds that 
$$ K2 \sim \log_{3/4}(\epsilon\gamma). $$
\end{lemma}
\begin{proof}
Using a Taylor expansion for $\ln(1-\epsilon/2)$ and $\ln(1-\sqrt{\frac{2}{\pi}}\gamma)$ in 0 we get

$$(1-\epsilon/2)^{\frac{1}{K-1}} = \exp\left(\frac{\ln(1-\epsilon/2)\ln(1-\sqrt{\frac{2}{\pi}}\gamma)}{\ln(\epsilon/2)-\ln(1-\sqrt{\frac{2}{\pi}}\gamma)}\right)
= \exp\left(-\alpha\right) $$

where 

$$ \alpha = \frac{1}{\sqrt{2\pi}}\frac{\epsilon\gamma}{\ln(\epsilon/2)-\ln(1-\sqrt{\frac{2}{\pi}}\gamma)} + o(\epsilon\gamma) \sim  \frac{1}{\sqrt{2\pi}}\frac{\epsilon\gamma}{\ln(\epsilon/2)-\ln(1-\sqrt{\frac{2}{\pi}}\gamma)}.$$

Using $\ln(1-e^{-x}) \underset{x\to 0}{\sim} \ln(x)$, it holds that

$$ K_2 =\log_{3/4}\left(1-e^{-\alpha}\right) \sim \log_{3/4}(\alpha) \sim \log_{3/4}(\epsilon\gamma). $$
\end{proof}

\begin{theorem}
\label{proof:hybrid_quantum_perceptron}
Let $S$ be a linearly separable sample of $N$ points of margin $\gamma$. Algorithm \QPercep finds a perfect separator with probability at least $1-\epsilon$ and has a complexity of
\begin{equation*}
     O\left(\frac{\sqrt{N}}{\gamma}\ln(1/\epsilon)\ln\left(\frac{1}{\gamma\epsilon}\right)\right)\; .
\end{equation*}
\end{theorem}
\begin{proof}
The algorithm can fail because of two reasons. It is possible that none of the hyperplanes $w_i$, $i=1,\ldots,K$, separate the classes and it is also 
possible that the quantum search gives a wrong result.\\
The exact value of $K$ we take is $K = \left\lceil\frac{\ln(\epsilon/2)}{\ln(1-\sqrt{2}\gamma/\sqrt{\pi})}\right\rceil  = O\left(\frac{\ln(1/\epsilon)}{\gamma}\right)$ because of lemma \ref{lemma:K2}. The probability that a randomly drawn hyperplane separates the data is $\sqrt{2/\pi}\gamma$ (from \citealp[Proof of theorem 2]{wiebe2016quantum}). Thus, the probability that at least one hyperplane separates the classes is 
\begin{align*}
    \P(\text{separating $w$ exists}) &= 1- \left(1-\sqrt{\frac{2}{\pi}}\gamma\right)^K
    \geq \left(1-\sqrt{\frac{2}{\pi}}\gamma\right)^{\frac{\ln(\epsilon/2)}{\ln(1-\sqrt{2}\gamma/\sqrt{\pi})}} = 1-\frac{\epsilon}{2}\; .
\end{align*}
Next we will assume that one of the $K$ hyperplanes separates the classes. The algorithm will still return a wrong answer if it identifies a non-separating hyperplane as a separating one. The worst case is when the separating hyperplane is the $K^{\text{th}}$ one. The probability that $K-1$ non-separating hyperplanes are all correctly identified is 
\begin{equation*}
    \left( 1 - \frac{3}{4}^{K_2}\right)^{K-1} \geq 1-\frac{\epsilon}{2}\; ,
\end{equation*}

where 
\begin{equation*}
    K_2 = \left\lceil\log_{3/4}\left(1-\left(1-\frac{\epsilon}{2}\right)^{\frac{1}{K-1}}\right)\right\rceil  = O\left(\ln(1/(\gamma\epsilon))\right)\; \text{(from lemma \ref{lemma:K2})}.
\end{equation*}

The probability of failure is then bounded by 
\begin{equation*}
    \P(\text{failure}) \leq \underbrace{\frac{\epsilon}{2}}_{\text{separating $w$ doesn't exist}} + \underbrace{\frac{\epsilon}{2}}_{\text{one non-separating hyperplane misidentified}} = \epsilon
\end{equation*}
and the complexity is 
\begin{equation*}
    O\left(K K_2 \sqrt{N}\right) = O\left(\frac{\sqrt{N}}{\gamma}\ln(1/\epsilon)\ln\left(\frac{1}{\gamma\epsilon}\right)\right)
\end{equation*}
which concludes the proof.
\end{proof}

\subsection{Proof of Theorem \ref{proof:generalization_quantum_perceptron}}
\label{sec:proof2}

For proving Theorem \ref{proof:generalization_quantum_perceptron}, the following definition and lemma are useful.

\begin{definition}
We define the Leave-one-out (LOO) error on a dataset $S$ by %
\begin{equation}
    \label{eq:rloo}
    \hat{R}_{LOO}(S) = \frac{1}{N}\sum_{i=1}^N{\1\{h_{S-\{x_i\}}(x_i)\neq y_i\}}\; ,
\end{equation}
where $h_{S-\{x_i\}}$ is the hypothesis returned by \QPercep on $S-\{x_i\}$, which is the same
as $S$ except that $x_i$ has been deleted.
\end{definition}

The lemma below shows the link between the expected risk and the Leave-one-out error.
\begin{lemma}[{\citealp[Lemma 5.3]{mohri2018foundations}}]\label{lemma:equal_risk}
For any $N\geq 1$,
\begin{equation*}
    \label{eq:equal_risks}
    \underset{S\sim \cD^N}{\E} \left[R(h_S)\right] = \underset{S'\sim \cD^{N+1}}{\E} [\hat{R}_{LOO}(S')]\; .
\end{equation*}
\end{lemma}

\begin{theorem}
\label{proof:generalization_quantum_perceptron}
Assume that the data is linearly separable. Let $h_S$ be the hypothesis returned by the
\QPercep algorithm after training over a sample $S$ of size $N$ drawn according to some distribution $\cD$. Then, the expected error of $h_S$ is bounded as follows:
\begin{equation*}
    \E_{S\sim \cD^N}\left(R(h_S)\right) \leq \sqrt{\frac{\pi}{2}}\frac{\log 1/\epsilon}{N+1}\E_{S\sim \cD^{N+1}}\left(\frac{1}{\gamma_S}\right)\; .
\end{equation*}
\end{theorem}
\begin{proof}
The proof is based on computing an upper bound of the Leave-one-out error. 
%
Since the hyperplanes are drawn beforehand, they are the same for all instances $(S-\{x_i\})_i, \forall i = 1,\ldots,N$. We also assume that there is at least one hyperplane that separates the training set $S$ of size $N$ (true with probability $1-\epsilon$). If $N\leq K$ then  the number of errors in $\hat{R}_{LOO}$ is naturally bounded by $N\leq K$ so it holds that $\hat{R}_{LOO}  \leq K/N$. Thus we can restrict ourselves to the non trivial case where $K<N$.

We know that there is an hyperplane that separates the training set $S$ correctly. Apart this hyperplane, noted $w_K$, the worst scenario is when the other ones all classify correctly all the data except one.
Without loss of generality we consider that each $w_k \text{ misclassifies only } x_k$, $\forall 1\leq i <K$.
So we will have one error for each of the $K-1$ first predictions. 
Now, when \QPercep is trained on $S-\{x_i\}$,  $\forall K\leq i \leq N$, the algorithm will choose the hyperplane $w_K$ because it is the only one that correctly separates  $S-\{x_i\}$ for $i=K,\ldots,N$.
Since $w_K$ is the hyperplane returned by \QPercep on all the sample $S$, it will also correctly classify the points $x_i$, $\forall K\leq i \leq N$. Hence 
it holds that 
\begin{equation*}
    \hat{R}_{LOO} \leq \frac{K}{N}\; .
\end{equation*}
Using Lemma \ref{lemma:equal_risk} and $K \sim \sqrt{\frac{\pi}{2}} \frac{\ln(1/\epsilon)}{\gamma}$ (lemma \ref{lemma:K}), we obtain
\begin{equation*}
    \E_{S\sim \cD^N}\left(R(h_S)\right) \leq \sqrt{\frac{\pi}{2}}\frac{\log 1/\epsilon}{N+1}\E_{S\sim \cD^{N+1}}\left(\frac{1}{\gamma_S}\right)\; .
\end{equation*}
\end{proof}

\bibliography{biblio}

\end{document}